\documentclass[11pt]{article}

\newcommand{\jmp}{J. Math. Phys.~}

\newcommand{\prl}{Phys. Rev. Lett.~}
\newcommand{\pra}{Phys. Rev. A~}
\newcommand{\pla}{Phys. Lett. A~}

\usepackage{young}
\usepackage[latin1]{inputenc}
\usepackage{appendix}
\usepackage{epstopdf}
\usepackage{xcolor}
\usepackage{subfigure}
\usepackage[T1]{fontenc}
\usepackage[sc]{mathpazo}
\usepackage{amsmath}
\usepackage{amssymb}
\usepackage{graphicx,color}
\usepackage{framed}
\usepackage{multirow}
\usepackage{enumerate}
\usepackage{amsthm}
\usepackage{amsfonts,mathrsfs}
\usepackage{geometry} 
\usepackage{eepic}
\usepackage{ifthen}
\usepackage[vcentermath]{youngtab}
\usepackage[unicode=true,pdfusetitle, bookmarks=true,bookmarksnumbered=false,bookmarksopen=false, breaklinks=false,pdfborder={0 0 0},backref=false,colorlinks=false] {hyperref}
\hypersetup{
colorlinks,linkcolor=myurlcolor,citecolor=myurlcolor,urlcolor=myurlcolor}
\definecolor{myurlcolor}{rgb}{0,0,0.7}

\usepackage{color}

\newcommand{\blue}{\textcolor{blue}}

\newcommand{\proj}[1]{| #1\rangle\!\langle #1 |}
\newcommand{\ketbra}[2]{|#1\rangle\!\langle#2|}

\usepackage{hyperref}
\hypersetup{pdfpagemode=UseNone}

\newcommand{\tinyspace}{\mspace{1mu}}

\newcommand{\op}[1]{\operatorname{#1}}

\newcommand{\abs}[1]{\left\lvert\tinyspace #1 \tinyspace\right\rvert}

\renewcommand{\det}{\operatorname{det}}
\renewcommand{\t}{{\scriptscriptstyle\mathsf{T}}}

\newcommand{\setft}[1]{\mathrm{#1}}
\newcommand{\lin}[1]{\setft{L}\left(#1\right)}
\newcommand{\density}[1]{\setft{D}\left(#1\right)}

\renewcommand{\vec}{\op{vec}}

\def \diag {\mathrm{diag}}

\def\I{\mathbb{1}}

\def\zero{\mathbf{0}}

\newenvironment{mylist}[1]{\begin{list}{}{
    \setlength{\leftmargin}{#1}
    \setlength{\rightmargin}{0mm}
    \setlength{\labelsep}{2mm}
    \setlength{\labelwidth}{8mm}
    \setlength{\itemsep}{0mm}}}
    {\end{list}}

\def\ot{\otimes}

\newcommand{\inner}[2]{\langle #1 , #2\rangle}

\newcommand{\out}[2]{| #1\rangle\langle #2 |}
\newcommand{\Inner}[2]{\left\langle #1 , #2\right\rangle}

\newcommand{\pa}[1]{(#1)}
\newcommand{\Pa}[1]{\left(#1\right)}

\newcommand{\set}[1]{\{#1\}}
\newcommand{\Set}[1]{\left\{#1\right\}}

\newcommand{\ket}[1]{|#1\rangle}

\newcommand{\braket}[1]{\langle#1\rangle}

\DeclareMathOperator{\trace}{Tr}
\newcommand{\ptr}[2]{\trace_{#1}\pa{#2}}
\newcommand{\Ptr}[2]{\trace_{#1}\Pa{#2}}
\newcommand{\tr}[1]{\ptr{}{#1}}
\newcommand{\Tr}[1]{\Ptr{}{#1}}

\newcommand{\Abs}[1]{\left|\tinyspace#1\tinyspace\right|}

\def\cB{\mathcal{B}}\def\cC{\mathcal{C}}\def\cE{\mathcal{E}}
\def\cI{\mathcal{I}}

\def\cT{\mathcal{T}}

\def\bG{\mathbf{G}}

\def\bsA{\boldsymbol{A}}\def\bsB{\boldsymbol{B}}\def\bsC{\boldsymbol{C}}\def\bsE{\boldsymbol{E}}
\def\bsF{\boldsymbol{F}}\def\bsG{\boldsymbol{G}}
\def\bsM{\boldsymbol{M}}
\def\bsP{\boldsymbol{P}}\def\bsT{\boldsymbol{T}}
\def\bsU{\boldsymbol{U}}\def\bsX{\boldsymbol{X}}\def\bsY{\boldsymbol{Y}}

\def\bsr{\boldsymbol{r}}

\def\T{\textsf{T}}

\def\bbC{\mathbb{C}}

\def\bbR{\mathbb{R}}

\def\sfU{\mathsf{U}}

\newtheorem{thrm}{Theorem}[section]

\newtheorem{prop}[thrm]{Proposition}
\newtheorem{cor}[thrm]{Corollary}
\theoremstyle{definition}
\newtheorem{definition}[thrm]{Definition}
\newtheorem{remark}[thrm]{Remark}
\newtheorem{exam}[thrm]{Example}

\numberwithin{equation}{section}

\newcounter{questionnumber}

\begin{document}

\title{Circulant quantum channels and its applications}

\author{\blue{Bing Xie}\footnote{E-mail: xiebingjiangxi2023@163.com}\quad\text{and}\quad \blue{Lin Zhang}\footnote{E-mail: godyalin@163.com}\\
  {\it\small School of Science, Hangzhou Dianzi University, Hangzhou 310018, PR~China}}
\date{}
\maketitle

\begin{abstract}
This note introduces a family of circulant quantum channels --- a
subclass of the mixed-permutation channels --- and investigates its
key structural and operational properties. We show that the image of
the circulant quantum channel is precisely the set of circulant
matrices. This characterization facilitates the analysis of
arbitrary $n$-th order Bargmann invariants. Furthermore, we prove
that the channel is entanglement-breaking, implying a substantially
reduced resource cost for erasing quantum correlations compared to a
general mixed-permutation channel. Applications of this channel are
also discussed, including the derivation of tighter lower bounds for
$\ell_p$-norm coherence and a characterization of its action in
bipartite systems.
\end{abstract}

\newpage

\section{Introduction}

A quantum channel is a mathematical description of quantum
information transmission or processing, represented by a completely
positive trace-preserving map (CPTP), which can characterize the
state evolution between quantum systems. Quantum channels play a
significant role in many physical processes, including quantum noise
\cite{qcqi}, quantum computing \cite{Takagi2020}, quantum
communication \cite{Divincenzo1995}, etc. Quantum systems are
affected by noise in actual environments, while quantum channels can
simulate the impact of environmental disturbances such as
decoherence \cite{Maziero2009} and amplitude damping
\cite{Khatri2020} on quantum systems. Quantum computing essentially
accomplishes information processing through the series and parallel
combination of quantum gates, where each quantum gate
\cite{Piveteau2025} corresponds to a definite channel. In quantum
communication, quantum channels are used to transmit or protect
quantum information. For instance, in quantum teleportation
\cite{Bennett1993}, quantum channels are utilized to transmit
unknown quantum states, and in quantum key distribution
\cite{Boyer2016}, encoded quantum states are transmitted through
quantum channels.

In \cite{Zhang2024}, the authors investigate a special type of
channel, the so-called \emph{mixed-permutation channel}, defined
with help of all permutations in $S_d$, the set of all permutations
of $d$ distinct elements. For any coherence measure, the coherence
measure of quantum states does not increase under the
mixed-permutation channel. Therefore, the mixed-permutation channel
can be used to estimate quantum coherence. In the above definition
of the mixed-permutation channel, they used \emph{all} permutations
in $S_d$. Instead, in this paper, we will use the
\emph{appropriate} subset of $S_d$ to define a special class of mixed
permutation channels. When all the coefficients are uniformly
distributed, this is a channel that converts any matrix into a
circulant matrix, which we call a circulant quantum channel. For
$\ell_p$-norm $(p=1,2)$ coherence \cite{Baumgratz}, the circulant
quantum channel can also be used to estimate coherence, and the
estimation effect is better than that of the mixed-permutation
channel, studied in \cite{Zhang2024}.

Circulant matrices play a vital role in characterizing Bargmann
invariants. As established in \cite{Fernandes2024}, these invariants
serve as witnesses for quantum imaginarity. Numerical evidence
strongly suggests that for four pure states, the boundary of the
Bargmann invariant set coincides precisely with that of circulant
Gram matrices. This equivalence was rigorously proven for
third-order invariants in \cite{Li2025}, which provided a complete
characterization for circulant Gram matrices of $n$ pure states.
Subsequently the researcher employed alternative methods to
characterize boundaries for third- and fourth-order invariants
\cite{Zhang2025}. These results support a natural conjecture: for
arbitrary $n$, the set of $n$-th order Bargmann invariants equals
the set of circulant Gram matrices. Notably, circulant quantum
channels provide a natural framework for resolving this
characterization problem, as they enable reconstruction of proof
methodologies like Xu and Pratapsi et al.'s
\cite{Xu2026PLA,Pratapsi2025}.

This work aims to investigate the fundamental properties of the
circulant quantum channel as well as its specific manifestation. In
addition, we employ such channel to reformulate the results for
characterization of Bargmann invariant sets and utilize it for
estimating quantum coherence. We demonstrate that a circulant
quantum channel can yield a tighter lower bound on quantum coherence
as measured by the $\ell_p$-norm coherence measure. Finally, we
extended the role of the circulant quantum channel to bipartite
systems. We find that when the circulant quantum channel acts
locally on a bipartite system, the entanglement of the bipartite
state is completely destroyed.

This paper is organized as follows. Section~\ref{sect:2prelimi}
presents the necessary background. In Section~\ref{sect:circq}, we
introduce circulant quantum channels and investigate their
properties. Applications of these channels are discussed in
Section~\ref{sect:appl}. Section~\ref{sect:bipartexp} extends the
notion of circulant channels to bipartite systems. Finally, we
present our conclusions in Section~\ref{sect:concl}.

\section{Preliminaries}\label{sect:2prelimi}

We consider quantum systems in the $d$-dimensional Euclidean space
$\bbC^d$ and fix the reference basis to $\set{\ket{i}}_{k=1}^d$. The
quantum state is described by a density matrix, i.e., a positive
semidefinite matrix of unit trace. We then denote the set of all
$d\times d$ complex matrices as $\lin{\bbC^d}$, and the set of all
density matrices as $\density{\bbC^d}$.

For a tuple of $n$ quantum states
$\Psi=(\rho_1,\rho_2,\cdots,\rho_n)$, its Bargmann invariant is
defined as $\Tr{\rho_1\rho_2\cdots\rho_n}$, which is invariant under
unitary transformation. To characterize the set of $n$-order
Bargmann invariants, circulant matrices are introduced
\cite{Fernandes2024}. In quantum systems of arbitrary finite
dimension, any permissible value of Bargmann invariants can be
achieved using pure states, which has a circular Gram matrix
\cite{Xu2026PLA}.

A $d\times d$ circulant matrix $\bsC$ is a type of matrix with a special structure, in which each row is a circulant shift of the previous row \cite{kra2012}, in the form of
\begin{eqnarray}
    \bsC=\begin{pmatrix}
        c_0 & c_{1} & c_{2} & \cdots & c_{d-1} \\
        c_{d-1} & c_0 & c_{1} & \cdots & c_{d-2} \\
        c_{d-2} & c_{d-1} & c_0 & \cdots & c_{d-3} \\
        \vdots & \vdots & \vdots & \ddots & \vdots \\
        c_{1} & c_{2} & c_{3} & \cdots & c_0
    \end{pmatrix}.
\end{eqnarray}
This means that at most only $d$ numbers are needed to determine a circulant matrix.

For $\pi\in S_d$, a permutation on the set $\set{1,\ldots,d}$, we
can define a permutation matrix $\bsP_\pi$
\begin{eqnarray}
    \bsP_\pi = \sum_{i=1}^d \out{\pi(i)}{i}.
\end{eqnarray}
By conventions, we denote $\bsP^0_{\pi}\equiv\I_d$, the identity
matrix on $\bbC^d$. It is easily seen \cite{Zhang2024} that
$\bsP_\pi \bsP_\sigma = \bsP_{\pi\sigma}$ and $\bsP_\pi^\dagger =
\bsP_\pi^\t = \bsP_{\pi^{-1}}=\bsP_\pi^{-1}$ for any $\pi$ and
$\sigma$ in $S_d$. Then the circulant matrix can be expressed as
$\bsC=\sum_{k=0}^{d-1} c_k \bsP^k_{\pi_0}$, where
$\pi_0=(d,\cdots,2,1)\in S_d$ and $\bsP^d_{\pi_0}=\I_d$.

In the following, we use the notation
\begin{eqnarray*}
    s\oplus t:=\begin{cases}
        s+t,&s+t\leqslant d\\
        s+t-d,&d< s+t
    \end{cases},~
    s\ominus t:=\begin{cases}
        s-t,&0\leqslant s-t\\
        s-t+d,& s+t<0
    \end{cases}.
\end{eqnarray*}
Under this notation, we have $\bsP_{\pi_0}^{k}=\sum_{i=1}^{d}\ketbra{i}{i\oplus k}$ for $k=0,\cdots,d$, due to
\begin{eqnarray*}
    \bsP_{\pi_0}^{k}=\sum_{i_1,\cdots,i_k=1}^{d}\braket{i_1|\pi_0(i_2)}\cdots\braket{i_{k-1}|\pi_0(i_k)}\ketbra{\pi_0(i_1)}{i_k}=\sum_{i_k=1}^{d}\ketbra{\pi^k_0(i_k)}{i_k}.
\end{eqnarray*}

\section{Circulant quantum channels}\label{sect:circq}

A quantum channel is a completely positive definite and
trace-preserving linear mapping that maps a density matrix to
another density matrix. All legal quantum channels have a Kraus
representation: $\Xi(\cdot) = \sum_\mu K_\mu (\cdot)
K_\mu^\dagger$, where $\sum_\mu K_\mu ^\dagger K_\mu = \I_d $.

Next, what we will do is the identification and systematic analysis
of circulant channels as a natural and operationally significant
subclass within quantum resource theories. While the broader class
of mixed-permutation channels has been introduced for tasks like
coherence estimation \cite{Zhang2024}, their restriction to the
cyclic subgroup---yielding the circulant structure---has not been
explicitly developed as an independent resource-theoretic object.
This gap is noteworthy because circulant channels possess unique
symmetries and algebraic properties that make them a potent tool for
investigating the geometric structure of quantum state space.

We now consider the following special mixed permutation channel and
explore its properties and applications.

\begin{definition}
Fix $\pi_0=(d,\cdots,2,1)\in S_d$. For any $d$-dimensional
probability vector
$\boldsymbol{\lambda}=(\lambda_0,\ldots,\lambda_{d-1})$, the
\emph{circulant quantum channel with weight $\boldsymbol{\lambda}$},
denoted by $\Phi_{\boldsymbol{\lambda}}$, is defined as
\begin{eqnarray}
\Phi_{\boldsymbol{\lambda}}(\bsX) := \sum_{k=0}^{d-1}
\lambda_k\bsP^k_{\pi_0}\bsX\bsP^{-k}_{\pi_0}.
\end{eqnarray}
In particular, $\Phi_{\boldsymbol{\lambda}}$ is abbreviated as
$\Phi$ when $\boldsymbol{\lambda}=(\frac1d,\ldots,\frac1d)$. That
is,
\begin{eqnarray}\label{phi}
\Phi(\bsX) :=
\frac1d\sum_{k=0}^{d-1}\bsP^k_{\pi_0}\bsX\bsP^{-k}_{\pi_0}.
\end{eqnarray}
This is called the \emph{circulant quantum channel} of uniform
weight.
\end{definition}

By definition, we have the following properties:
\begin{prop}
It holds that
\begin{enumerate}[(i)]
\item $\Phi_{\boldsymbol{\lambda}}$ is unital, that is $\Phi_{\boldsymbol{\lambda}}(\I_d) = \I_d$.
\item $\Phi_{\boldsymbol{\lambda}}$ is not self-adjoint with respect to Hilbert-Schmidt inner product.
\item $[\Phi_{\boldsymbol{\lambda}}(\bsX )]^\t = \Phi_{\boldsymbol{\lambda}}(\bsX^\t)$.
\item $\Phi_{\boldsymbol{\lambda}}$ is Hermitian-preserving, that is, $\Phi_{\boldsymbol{\lambda}}(\bsX)=\Phi_{\boldsymbol{\lambda}}(\bsX)^\dagger$ for $\bsX=\bsX^\dagger$.
\item $\Phi \circ \Phi = \Phi$, that is $\Phi (\Phi(\bsX)) = \Phi(\bsX)$ for any matrix $\bsX$ acting on
$\bbC^d$.
\end{enumerate}
\end{prop}

\begin{proof}
(1) Since $\Phi_{\boldsymbol{\lambda}}(\I_d)=\sum_{k=0}^{d-1}\lambda_k \bsP^k_{\pi_0} \bsP_{\pi_0}^{-k}=\I_d$, $\Phi_{\boldsymbol{\lambda}}$ is unital.\\
(2) Since
\begin{eqnarray*}
\inner{\bsX }{\Phi_{\boldsymbol{\lambda}}(\bsY)}&=&\inner{\bsX
}{\sum_{k=0}^{d-1}\lambda_k \bsP^k_{\pi_0}\bsY
\bsP_{\pi_0}^{-k}}=\inner{\sum_{k=0}^d\lambda_k
\bsP^{-k}_{\pi_0}\bsX
\bsP_{\pi_0}^{k}}{\bsY}=\inner{\Phi_{\boldsymbol{\lambda}}^\dagger(\bsX
)}{\bsY},
\end{eqnarray*}
we have $\Phi_{\boldsymbol{\lambda}}^\dagger(\bsX )=\sum_{k=0}^{d-1}\lambda_k
\bsP^{-k}_{\pi_0}\bsX  \bsP_{\pi_0}^{k}\neq \Phi_{\boldsymbol{\lambda}}(\bsX )$. \\
(3)$[\Phi_{\boldsymbol{\lambda}}(\bsX)]^\t = \Pa{
\sum_{k=0}^{d-1}\lambda_k \bsP^k_{\pi_0} \,\bsX  \,
\bsP^{-k}_{\pi_0}}^\t =  \sum_{k=0}^{d-1} \lambda_k \bsP^{k}_{\pi_0}
\,
\bsX ^\t \, \bsP^{-k}_{\pi_0} = \Phi_{\boldsymbol{\lambda}}(\bsX ^\t)$. \\
(4)
$[\Phi_{\boldsymbol{\lambda}}(\bsX )]^\dagger = \left(  \sum_{k=0}^{d-1}\lambda_k
\bsP^k_{\pi_0} \,\bsX  \, \bsP^{-k}_{\pi_0} \right)^\dagger =
\sum_{k=0}^{d-1} \lambda_k\bsP^{k}_{\pi_0} \, \bsX ^\dagger \,
\bsP^{-k}_{\pi_0} = \Phi_{\boldsymbol{\lambda}}(\bsX )$ for $\forall \bsX \in
Herm(\mathbb{C}^d)$.\\
(5) Indeed,
\begin{eqnarray*}
    \Phi \circ
    \Phi(\bsX)&=&\frac{1}{d^2}\sum_{p,q=0}^{d-1}\bsP_{\pi_0}^{p+q}\bsX
    \bsP_{\pi_0}^{-(p+q)}
\end{eqnarray*}
Let $p+q=k$ and denote by $N(k)$ the number of occurrences of $k$,
\begin{eqnarray}\label{N}
    N(k)=\begin{cases}
        k+1,&0\leqslant k\leqslant d-1\\
        2d-1-k,&d\leqslant k\leqslant 2(d-1)
    \end{cases}.
\end{eqnarray}
Due to $\bsP_{\pi_0}^d=\bsP_{\pi_0}^0$, we have $\Phi \circ
\Phi(\bsX)=\frac{1}{d^2}\sum_{p,q=1}^{d}\bsP_{\pi_0}^{p\oplus q}\bsX
\bsP_{\pi_0}^{-(p\oplus q)}$.
Thus, we obtain $p\oplus q=k'\in\{1,\cdots,d\}$, and $N(k')=d$ for
$\forall k'$. Moreover, it holds that
\begin{eqnarray*}
\Phi \circ
\Phi(\bsX)&=&\frac{1}{d^2}\sum_{p,q=1}^{d}\bsP_{\pi_0}^{p\oplus
q}\bsX \bsP_{\pi_0}^{-(p\oplus
q)}=\frac{1}{d}\sum_{k'=1}^{d}\bsP_{\pi_0}^{k'}\bsX
\bsP_{\pi_0}^{-k'} =\Phi(\bsX).
\end{eqnarray*}
We are done.
\end{proof}

Here, we present the expression of the image of channel $\Phi_\lambda$.
\begin{thrm}
For any matrix $\bsX\in\lin{\bbC^d}$, the corresponding outputting
image $\Phi_{\boldsymbol{\lambda}}(\bsX)$ is
\begin{eqnarray*}
\Phi_{\boldsymbol{\lambda}}(\bsX)=\sum_{i,j=1}^{d}\Tr{\bsP_{\pi_0}^{-(j-1)}\Lambda\bsP_{\pi_0}^{i-1}\bsX}\ketbra{i}{j},
\end{eqnarray*}
where $\Lambda:=\diag(\lambda_0,\cdots,\lambda_{d-1})$ is a diagonal
matrix formed by
$\boldsymbol{\lambda}=(\lambda_0,\lambda_1,\ldots,\lambda_{d-1})$.
\end{thrm}
\begin{proof} We first decompose matrix $\bsX$ into $\bsX=\sum_{m,n=1}^{d}x_{m,n}\out{m}{n}$, then
\begin{eqnarray*}
\Phi_{\boldsymbol{\lambda}}(\bsX) &=& \sum_{k=0}^{d-1}\lambda_k\bsP^k_{\pi_0}\bsX\bsP^{-k}_{\pi_0}=\sum_{k=0}^{d-1} \sum_{i,j,m,n=1}^{d} \lambda_k x_{m,n} \braket{i\oplus k|m}\braket{n|j\oplus k}\out{i}{j}\\
&=&\sum_{k=0}^{d-1} \sum_{i,j,m,n=1}^{d} \lambda_k x_{m,n} \braket{(i-1)\oplus (k+1)|m}\braket{n|(j-1)\oplus (k+1)}\out{i}{j}\\
&=& \sum_{i,j=1}^{d} \Tr{\sum_{k=0}^{d-1}\sum_{m,n=1}^{d}\lambda_k x_{m,n} \out{(k+1)\oplus(j-1)}{(k+1)\oplus(i-1)}\out{m}{n} }\out{i}{j}\\
&=& \sum_{i,j=1}^{d}\Tr{\sum_{p=1}^{d}\out{p\oplus(j-1)}{p}\sum_{k=0}^{d-1}\lambda_k\out{k+1}{k+1}\sum_{q=1}^{d}\out{q}{q\oplus(i-1)}\sum_{m,n=1}^{d}x_{m,n}\out{m}{n}}\out{i}{j}\\
&=&\sum_{i,j=1}^{d}\Tr{\bsP_{\pi_0}^{-(j-1)}\Lambda\bsP_{\pi_0}^{i-1}\bsX}\out{i}{j}.
\end{eqnarray*}
This completes the proof.
\end{proof}

When $\boldsymbol{\lambda}=(\frac1d,\cdots,\frac1d)$, This mixed
permutation channel simplifies to a circulant quantum channel --- a
more restricted class that maps any input matrix to a circulant
matrix.

\begin{cor}\label{t1}
For any inputting matrix $\bsX\in\lin{\bbC^d}$, the corresponding
outputting image $\Phi(\bsX)$ under the circulant quantum channel
$\Phi$ is a circulant matrix
\begin{eqnarray}
\Phi(\bsX)=\sum^{d-1}_{k=0}c_k(\bsX)\bsP_{\pi_0}^k,
\end{eqnarray}
where $c_k(\bsX):=\frac1d\Tr{\bsP_{\pi_0}^{-k}\bsX}$. Moreover,
$[\Phi(\bsX),\Phi(\bsY)]=\zero$ for any $\bsX$ and $\bsY$ acting on
$\bbC^d$.
\end{cor}

\begin{proof}
Due to $\boldsymbol{\lambda}=(\frac1d,\cdots,\frac1d)$, i.e.,
$\Lambda=\frac1d\I_d$, we have
\begin{eqnarray*}
\Phi(\bsX)&=&\sum_{i,j=1}^{d}\frac1d\Tr{\bsP_{\pi_0}^{-(j-1)}\bsP_{\pi_0}^{i-1}\bsX}\out{i}{j}\\
&=&\sum_{i,j=1}^{d}\frac1d\Tr{\bsP_{\pi_0}^{-(j\ominus i)}\bsX}\out{i}{j}\\
&=&\sum_{k=0}^{d-1}\sum_{i=1}^{d}\frac1d\Tr{\bsP_{\pi_0}^{-k}\bsX}\out{i}{i\oplus k}\\
&=&\sum^{d-1}_{k=0}\frac1d\Tr{\bsP_{\pi_0}^{-k}\bsX}\bsP_{\pi_0}^k.
\end{eqnarray*}
Hence the result.
\end{proof}


If $d=2$, then
$\Phi(\bsX)=\frac12(\tr{\bsX}\I+\tr{\sigma_1\bsX}\sigma_1)$, where
$\sigma_1$ is a Pauli matrix. It is consistent with the
mixed-permutation channel $\Delta(\bsX) := \frac{1}{d!} \sum_{\pi
\in S_d} \bsP_\pi \bsX \bsP_\pi^\dagger$\cite{Zhang2024}. When the
input matrix $\bsX$ takes the quantum state $\rho\in
\density{\bbC^d}$, the output state can be characterized
analytically as follows. Since all density matrices are Hermitian,
$c_m(\rho)$ must be a real number for $d=2m$.

\begin{cor}
For any quantum state $\rho = \sum_{i,j=1}^{d} \rho_{ij} \out{i}{j}
\in \density{\bbC^d}$, the corresponding output state $\Phi(\rho)$
under the channel is given by
\begin{eqnarray}
        \Phi(\rho)=\sum^{d-1}_{k=0}c_k(\rho)\bsP_{\pi_0}^k,
\end{eqnarray}
where $c_0=\frac1d$ and $c_k(\rho)=\frac1d\sum_{i=1}^{d}
\rho_{i,i\oplus k}$ satisfy $c_r(\rho)=\overline{c}_{d-r}(\rho)$ for
$r=1,\cdots,d-1$.
\end{cor}

For every channel $\Xi:\lin{\bbC^d}\rightarrow \lin{\bbC^d}$, there
must exist a nature representation $K(\Xi)\in \lin{\bbC^{d^2}}$ such
that $\vec(\Xi(\bsX))=K(\Xi)\vec(\bsX)$, where
$(K(\Xi))_{ij}=\tr{\bsE_i^\dagger\Xi(\bsE_j)}$ for a given set of
matrix bases $\set{\bsE_i\in \lin{\bbC^d}}$ \cite{Wolf2010}.
Moreover, by the Klaus representation of $\Xi$, we can write the
linear superoperator as  $K(\Xi)=\sum_{\mu}K_\mu\ot
\overline{K}_\mu$ \cite{Kukulski2021}. The spectrum of channel $\Xi$
is defined as a list of roots of the characteristic polynomial
$\det\Pa{K(\Xi)-\lambda\I}$, where each root appears according to
its degeneracy degree \cite{Wolf2010}. If $\lambda=1$ is an
eigenvalue of $\Xi$, then all operators $\bsX$ that satisfy
$\Xi(\bsX)=\lambda\bsX=\bsX$ constitute the set of fixed points of
$\Xi$.

\begin{cor}\label{fixpoint}
The set of fixed points of the circulant quantum channel $\Phi$,
defined in Eq.~\eqref{phi}, is precisely the set of all circulant
matrices.
\end{cor}
\begin{proof}
From the Corollary \ref{t1}, we have
$(\Phi(\bsX))_{p,q}=\sum_{k=0}^{d-1}x_{p\oplus k,q\oplus k}$ for any
matrix $\bsX$. Let $\bsX$ is the fixed point of $\Phi$, then
$x_{p,q}=(\Phi(\bsX))_{p,q}$. It holds that
\begin{eqnarray*}
x_{i,i\oplus k}=\frac1d\sum_{j=1}^{d}x_{j,j\oplus k} ~\text{for}
~\forall i\in\{1,2,\cdots,d\},~\forall k\in\{0,1,\cdots,d-1\}.
\end{eqnarray*}
Thus, $x_{i,i\oplus k}=x_{j,j\oplus k}$ for all $i,j$. That is,
$\bsX$ is a circulant matrix. On the other hand, If
$\bsX=\sum_{k=0}^{d-1}a_k\bsP_{\pi_0}^k$ is a circulant matrix, then
$\Phi(\bsX)=\sum_{k=0}^{d-1}c_k(\bsX)\bsP_{\pi_0}^k=\bsX$, where
$c_k(\bsX)=\frac1d\sum_{i=1}^{d}x_{i,i\oplus k}=a_k$.
\end{proof}

By calculation, we have
$K(\Phi)=\frac1d\sum_{k=0}^{d-1}\bsP^k_{\pi_0}\ot\bsP^{k}_{\pi_0}$.
We diagonalize $\bsP_{\pi_0}$ as
\begin{eqnarray}\label{eq:Omega}
\bsP_{\pi_0}=\bsF\Omega\bsF^\dagger, \quad
\Omega:=\diag(1,\omega,\cdots,\omega^{d-1})
\end{eqnarray}
where $\omega:=\exp\Pa{\frac{2\pi\mathrm{i}}d}$ and $\bsF$ is the
discrete Fourier transform (DFT) matrix \cite{Gray2006}. Then,
$$
K(\Phi)=(\bsF\ot\bsF)\Pa{\frac1d\sum_{k=0}^{d-1}\Omega^k\ot\Omega^{k}}(\bsF\ot\bsF)^\dagger,
$$
where
\begin{eqnarray}\label{DD}
\frac1d \sum^{d-1}_{k=0}\Omega^k\ot\Omega^k =
\sum^{d-1}_{\stackrel{i,j=0}{i+j\equiv0(\!\!\!\mod d)}}\proj{ij}.
\end{eqnarray}
Note that $K(\Phi)$ is locally unitary equivalent to Eq.~\eqref{DD}.
Thus, from Eqs.~\eqref{N} and \eqref{DD}, we obtain the spectrum of
$K(\Phi)$ is $\set{1_{(d)},0_{(d^2-d)}}$, which is also the spectrum
of the channel $\Phi$. This means that the dimension of the fixed
point space of channel $\Phi$ is $d$. Corollary~\ref{fixpoint}
indicates that this fixed point space is precisely the set of
circulant matrices.


\section{Applications of the circulant quantum channel}\label{sect:appl}

Leveraging the specific expression of the circulant quantum channel,
we demonstrate its utility in characterizing Bargmann invariants and
quantifying quantum coherence. As an illustrative example, we show
that these channels yield a tighter lower bound for the
$\ell_p$-norm ($p=1,2$) measure of coherence.

\subsection{Restatement of the relevant results of Bargmann invariants}

For $n$-tuple of quantum pure states
$\Psi=(\ket{\psi_1},\cdots,\ket{\psi_n})$ on $\bbC^d$, the Gram
matrix of $\Psi$ is $\bsG(\Psi):=(\Inner{\psi_i}{\psi_j})_{n}$.
Clearly $\bsG(\Psi)$ is positive semidefinite with principal
diagonal entries being one. The Bargmann invariant corresponding to
$\Psi$ is
\begin{eqnarray*}
\Tr{\psi_1\cdots\psi_n}=\prod_{i=1}^n\inner{\psi_i}{\psi_{i\oplus
1}}=\prod_{i=1}^n\bsG(\Psi)_{i,i\oplus1},
\end{eqnarray*}
where $\psi\equiv\proj{\psi}$. By using the circulant channel
$\Phi$, we can reproduce the results in \cite{Xu2026PLA}. For
completeness, we reformulate these results and their proofs in our
language.

\begin{cor}[\cite{Xu2026PLA}]\label{prop:phaseinvariant}
Given any $n$-tuple of complex unit vectors in $\bbC^d$,
$\Psi=(\ket{\psi_1},\ldots,\ket{\psi_n})$ with
$\Tr{\psi_1\cdots\psi_n}\neq0$, where $\psi_k\equiv\proj{\psi_k}$,
there exists a $n$-tuple of unit vectors
$\widetilde\Psi=(\ket{\tilde\psi_1},\ldots,\ket{\tilde\psi_n})$ such
that the Gram matrix
$G(\widetilde\Psi)=(\inner{\tilde\psi_i}{\tilde\psi_j})_{n\times n}$
is a circular matrix with
\begin{enumerate}[(i)]
\item
$\inner{\tilde\psi_1}{\tilde\psi_2}=\inner{\tilde\psi_2}{\tilde\psi_3}=\cdots=\inner{\tilde\psi_{n-1}}{\tilde\psi_n}=\inner{\tilde\psi_n}{\tilde\psi_1}$,
\item $\arg \Tr{\psi_1\cdots\psi_n}=\arg
\Tr{\tilde\psi_1\cdots\tilde\psi_n}$,
\item $\Abs{\Tr{\psi_1\cdots\psi_n}}\leqslant
\Abs{\Tr{\tilde\psi_1\cdots\tilde\psi_n}}$, where $\arg
z\in[0,2\pi)$ is the principle argument of the complex number $z$.
\end{enumerate}
\end{cor}

\begin{proof}
We first clearly demonstrate how to adjust the phases of the
inner-product factors in the $n$th order Bargmann invariant
$\Inner{\psi_1}{\psi_2}\Inner{\psi_2}{\psi_3}\cdots
\Inner{\psi_n}{\psi_1}=re^{\mathrm{i}\theta}$: we transform the $n$
factors $\Inner{\psi_k}{\psi_{k\oplus1}}=r_ke^{\mathrm{i}\theta_k}$
for $k\in\set{1,\ldots,n}$, originally carrying distinct phases
$e^{\mathrm{i}\theta_k}$, into a new set of factors
$\Inner{\psi'_k}{\psi'_{k\oplus1}}=r_ke^{\mathrm{i}\frac\theta n}$
where each possesses a common, uniform phase
$e^{\mathrm{i}\frac\theta n}$. Subsequently, building on this phase
alignment, we apply a circulant quantum channel to equalize the
moduli of all inner-product factors while maintaining their uniform
phase. This step transforms the corresponding Gram matrix $G(\Psi')$
into a circulant matrix, which, by construction, is precisely the
Gram matrix of the target tuple  $\tilde\Psi$. The proof then
continues as follows.

Let $\Tr{\psi_1\cdots\psi_n}=re^{\mathrm{i}\theta}$ with
$r=\abs{\Tr{\psi_1\cdots\psi_n}}>0$ and $\theta\in[0,2\pi)$. Let
$\Inner{\psi_k}{\psi_{k\oplus1}}=r_ke^{\mathrm{i}\theta_k}$, where
$r_k=\abs{\Inner{\psi_k}{\psi_{k\oplus1}}}>0$ and
$\theta_k\in[0,2\pi)$ for $k\in\set{1,2,\ldots,n}$. Thus
$r=r_1r_2\cdots r_n$.\\
\textbf{Step 1: Based on $\Psi$, we define a new tuple of unit
vectors $\Psi'=(\ket{\psi'_1},\ldots,\ket{\psi'_n})$}. In fact,
$\ket{\psi'_k}=\bsU\ket{\psi_k}$ for some unitary $\bsU\in\sfU(d)$
if and only if there exists a diagonal unitary
$\bsT\in\sfU(1)^{\times n}$ such that $G(\Psi')=\bsT^\dagger
G(\Psi)\bsT$. Thus we can choose suitable
$\bsT=\diag(e^{\mathrm{i}\alpha_1},\ldots,e^{\mathrm{i}\alpha_n})$,
and construct $\Psi'=(\ket{\psi'_1},\ldots,\ket{\psi'_n})$ by
defining $\ket{\psi'_k} :=e^{\mathrm{i}\alpha_k}\ket{\psi_k}$, where
$k\in\set{1,\ldots,n}$. We require the new tuple $\Psi'$ to satisfy
the condition that all consecutive inner products have the same
phase:
\begin{eqnarray*}
\arg\Inner{\psi'_1}{\psi'_2}=\arg\Inner{\psi'_2}{\psi'_3}=\cdots=\arg\Inner{\psi'_{n-1}}{\psi'_n}=\arg\Inner{\psi'_n}{\psi'_1}=\frac\theta
n
\end{eqnarray*}
which leads to the following constraints on the phase angles
$\set{\alpha_k\mid k=1,\ldots,n}\subset\bbR$:
$$
\alpha_j=\alpha_1+(j-1)\frac\theta n-\sum^{j-1}_{k=1}\theta_k,\quad
j\in\set{2,\ldots,n}.
$$
In what follows, we explain about why we can do this! That is, we
can show that the existence of $\alpha_k$'s satisfying the
constraints. Indeed, we rewrite
$\Tr{\psi_1\cdots\psi_n}=re^{\mathrm{i}\theta}$ in two ways:
\begin{eqnarray*}
\begin{cases}
re^{\mathrm{i}\theta}=(r_1e^{\mathrm{i}\frac\theta n})\cdots (r_n
e^{\mathrm{i}\frac\theta n}),\\
re^{\mathrm{i}\theta}=(r_1e^{\mathrm{i}\theta_1})\cdots
(r_ne^{\mathrm{i}\theta_n}).
\end{cases}
\end{eqnarray*}
What we want is that $r_ke^{\mathrm{i}\frac\theta n}$ will be
identified with
$r_ke^{\mathrm{i}\theta_k}=\Inner{\psi_k}{\psi_{k\oplus1}}$ \emph{up
to a phase factor}. This leads to the definition $\ket{\psi'_k}
:=e^{\mathrm{i}\alpha_k}\ket{\psi_k}$ for some $\alpha_k\in\bbR$
\emph{to be determined}. In such definition, we hope that
$r_ke^{\mathrm{i}\frac\theta n}=\inner{\psi'_k}{\psi'_{k\oplus1}}$.
Assuming this, we get that
\begin{eqnarray*}
r_ke^{\mathrm{i}\frac\theta n}&=&\inner{\psi'_k}{\psi'_{k\oplus1}}=
e^{\mathrm{i}(\alpha_{k\oplus1}-\alpha_k)}\inner{\psi_k}{\psi_{k\oplus1}}\\
&=&
e^{\mathrm{i}(\alpha_{k\oplus1}-\alpha_k)}r_ke^{\mathrm{i}\theta_k}=r_ke^{\mathrm{i}(\alpha_{k\oplus1}-\alpha_k+\theta_k)}.
\end{eqnarray*}
Under the assumption about the existence of
$\set{\alpha_k}^n_{k=1}\subset\bbR$, we derive the constraints: The
requirement concerning unknown constants $\alpha_k$'s are determined
by
\begin{eqnarray*}
\frac\theta n=\alpha_{k\oplus1}-\alpha_k+\theta_k
\Longleftrightarrow \alpha_{k\oplus1}-\alpha_k=\frac\theta
n-\theta_k,
\end{eqnarray*}
implying that
\begin{eqnarray*}
\alpha_{j\oplus1}=\alpha_1+(j-1)\frac\theta
n-\sum^{j-1}_{k=1}\theta_k,\quad j\in\set{2,\ldots,n}.
\end{eqnarray*}
Note that $\alpha_1$ is chosen freely. Let
\begin{eqnarray*}
g_j\equiv g_j(\theta,\theta_1,\ldots,\theta_n):=(j-1)\frac\theta
n-\sum^{j-1}_{k=1}\theta_k,\quad j\in\set{2,\ldots,n}.
\end{eqnarray*}
Thus
$\bsT=e^{\mathrm{i}\alpha_1}\diag\Pa{1,e^{\mathrm{i}g_2},\ldots,e^{\mathrm{i}g_n}}$.
Thus $\Psi'$ can be defined reasonably by using such constants
$\alpha_k$'s satisfying the mentioned constraints. In this
situation, $G(\Psi')=\bsT^\dagger G(\Psi)\bsT$ for
$\bsT=\diag(e^{\mathrm{i}\alpha_1},\ldots,e^{\mathrm{i}\alpha_n})\in\sfU(1)^{\times
n}$. \\
\textbf{Step 2: Based on $\Psi'$, we define a new tuple of wave
functions
$\tilde\Psi=(\ket{\tilde\psi_1},\ldots,\ket{\tilde\psi_n})$}. By
acting the \emph{circulant quantum channel} $\Phi$ on the Gram
matrix $G(\Psi')$ for the constructed tuple of unit vectors $\Psi'$
in step 1, we have that $\zero\leqslant\Phi(G(\Psi'))\in\cC_n$. Thus
there exists some tuple of unit vectors
$\widetilde\Psi=(\ket{\tilde\psi_1},\ldots,\ket{\tilde\psi_n})$ such
that $G(\widetilde\Psi) = \Phi(G(\Psi'))$ \cite{chefles2004}, and
\begin{eqnarray*}
\inner{\tilde\psi_j}{\tilde\psi_{j\oplus1}}&=&[G(\widetilde\Psi)]_{j,j\oplus1}=[\Phi(G(\Psi'))]_{j,j\oplus1}\\
&=&\frac1n\sum^n_{k=1}\inner{\psi'_k}{\psi'_{k\oplus1}}=\Pa{\frac1n\sum^n_{k=1}r_k}e^{\mathrm{i}\frac\theta
n},\quad\forall j\in\set{1,2,\ldots,n}.
\end{eqnarray*}
By the Arithmetic-Geometric Mean (AM-GM) Inequality, we have
\begin{eqnarray*}
&&\Abs{\Tr{\psi_1\cdots\psi_n}} =
\Abs{\prod^n_{k=1}\Inner{\psi_k}{\psi_{k\oplus1}}} =
\prod^n_{k=1}r_k\\
&&\leqslant \Pa{\frac1n\sum^n_{k=1}r_k}^n =
\Abs{\prod^n_{k=1}\Inner{\tilde\psi_k}{\tilde\psi_{k\oplus1}}}=
\Abs{\Tr{\tilde\psi_1\cdots\tilde\psi_n}}.
\end{eqnarray*}
In addition, we also see that
\begin{eqnarray*}
\begin{cases}
\Tr{\psi_1\cdots\psi_n}
=\Pa{\prod^n_{k=1}r_k}e^{\mathrm{i}\theta},\\
\Tr{\tilde\psi_1\cdots\tilde\psi_n}
=\Pa{\frac1n\sum^n_{k=1}r_k}^ne^{\mathrm{i}\theta}.
\end{cases}
\end{eqnarray*}
This completes the proof.
\end{proof}

We should remark here that, for the logical integrity of our work,
we reconstructed Xu's description of the Bargmann invariant set
within the framework of circulant quantum channels. Based on
Corollary~\ref{prop:phaseinvariant}, it is easily shown that
$\cB_n=\cB_{n|\mathrm{circ}}$, where $\cB_n$ is the set of all $n$th
order Bargmann invariants, and $\cB_{n|\mathrm{circ}}(\subset
\cB_n)$ is the subset of $n$th Bargmann invariants resulted from
circulant Gram matrices which is formed by special tuples of pure
states \cite{Xu2026PLA}. Indeed, for $z\in\cB_n$, if $z=0$, then
apparently $z=0\in\cB_{n|\mathrm{circ}}$; if $z\neq0$, which can be
realized as $z=\Tr{\psi_1\cdots\psi_n}$ for an $n$-tuple of unit
vectors $\Psi=(\ket{\psi_1},\ldots,\ket{\psi_n})$. Thus there exists
$\tilde z=\Tr{\tilde\psi_1\cdots\tilde\psi_n}$ by
Corollary~\ref{prop:phaseinvariant} such that
\begin{eqnarray*}
z=\frac{\prod^n_{k=1}r_k}{\Pa{\frac1n\sum^n_{k=1}r_k}^n}\tilde
z\in\cB_n|_{\mathrm{circ}}
\end{eqnarray*}
While this recasts known findings instead of presenting new ones, it
provides a different perspective that underscores the utility of
circulant quantum channels in characterizing these invariant sets.
In the following, we will give applications of the circulant quantum
channel to estimating coherence.

\subsection{Applying the circulant channel to estimate quantum coherence}

Given the reference basis $\set{\ket{i}}_{i=1}^d$, if the density
matrix $\rho$ of a quantum state is diagonal under this set of
bases, $\rho=\sum_{i=1}^{d}p_i\out{i}{i}$, then the state is called
an \emph{incoherent state}. The set of all incoherent states is
denoted by $\cI(\bbC^d)$. The so-called \emph{incoherent operations}
$\cE(*)=\sum_\nu K_\nu (*)K^\dagger_\nu$ are quantum channels that
satisfy condition $\frac{K_{\nu}\delta
K_{\nu}^{\dagger}}{\Tr{K_{\nu}\delta
K_{\nu}^{\dagger}}}\in\cI(\bbC^d)$ for all $\delta\in\cI(\bbC^d)$
and for all $\nu$, converting incoherent states into incoherent
states.

For the circulant quantum channel, we find that it maps any
incoherent state to the maximum mixed state.

\begin{cor}
$\Phi(\rho)$ is the maximum mixed state $\frac1d\I_d$ for any
incoherent states $\rho = \sum_{i=1}^{d} \lambda_{i} \proj{i}\in
\cI(\bbC^d)$.
\end{cor}

\begin{proof}
Since $\bsP_{\pi_0}^k\rho \bsP_{\pi_0}^{-k}=\sum_{i=1}^{d}
\lambda_{i}\ketbra{\pi_0^k(i)}{\pi_0^k(i)}\in
\mathcal{I}(\mathbb{C}^d)$, it follows that
$$
\Phi(\rho)=\frac{1}{d}\sum_{k=0}^{d-1}\sum_{i=1}^{d}
\lambda_{i}\ketbra{\pi_0^k(i)}{\pi_0^k(i)}=\frac1d\I_d.
$$
We are done.
\end{proof}

For every Kraus operator $\frac{\bsP^k_{\pi_0}}{\sqrt{d}}$ and any
incoherent states
$\delta=\sum_{i}^{d}\lambda_i\proj{i}\in\cI(\bbC^d)$, it holds that
$$
\frac{\bsP^k_{\pi_0}}{\sqrt{d}}\delta\frac{\bsP^{-k}_{\pi_0}}{\sqrt{d}}/\Tr{\frac{\bsP^k_{\pi_0}}{\sqrt{d}}\delta\frac{\bsP^{-k}_{\pi_0}}{\sqrt{d}}}=\sum_{i}^{d}\lambda_i\proj{\pi_0(i)}\in\cI(\bbC^d),
$$
so $\Phi$ is a incoherent operation. For any coherence measure $C$
and any quantum state $\rho \in \density{\bbC^d}$, the quantum
coherence of $\rho$ is bounded from below by the quantum coherence
of $\Phi(\rho)$, i.e. $  C(\rho) \geqslant C(\Phi(\rho))$.

Considering the $\ell_1$-norm coherence of the quantum state
$\rho=\sum_{i,j}\rho_{ij}\,|i\rangle\!\langle j|
$ is the sum of the magnitudes of all the off-diagonal entries:
\begin{eqnarray}\label{eq:cl1}
C_{\ell_1}(\rho):=\sum_{i\neq j}\abs{\rho_{ij}}.
\end{eqnarray}
Unlike the $\ell_1$-norm coherence, the $\ell_2$-norm coherence is defined as the sum of the squares of the magnitudes of the off-diagonal elements:
\begin{eqnarray}\label{eq:cl2}
    C_{\ell_2}(\rho):=\sum_{i\neq j}\abs{\rho_{ij}}^2.
\end{eqnarray}

The circulant quantum channel $\Phi$ can also be used to estimate
quantum coherence, providing a tighter lower bound for the
$\ell_1$-norm coherence and the $\ell_2$-norm coherence than for the
mixed-permutation channel $\Delta(\cdot)=\frac1{d!}\sum_{\pi\in
S_d}\bsP_{\pi}(\cdot)\bsP^\dagger_{\pi}$ \cite{Zhang2024}.

\begin{cor}
For any quantum state $\rho \in \density{\bbC^d}$, the $\ell_p$-norm
$(p=1,2)$ coherence of out states satisfy
\begin{eqnarray}
C_{\ell_p}(\rho) \geqslant C_{\ell_p}(\Phi(\rho))\geqslant
C_{\ell_p}(\Delta(\rho)).
\end{eqnarray}
\end{cor}

\begin{proof}
By definition, we have
\begin{eqnarray*}
C_{\ell_1}(\Phi(\rho))=\sum_{k=1}^{d-1}\Abs{\sum_{i=1}^{d}\rho_{i,i\oplus
k}},\quad  C_{\ell_1}(\Delta(\rho))=\Abs{\sum_{i\neq j}\rho_{i,j}}.
\end{eqnarray*}
Then by the triangle inequality, it holds that
\begin{eqnarray*}
C_{\ell_1}(\rho) \geqslant C_{\ell_1}(\Phi(\rho))\geqslant
C_{\ell_1}(\Delta(\rho)).
\end{eqnarray*}
Indeed, $C_{\ell_1}(\Phi(\rho))=\sum_{k=1}^{d-1}|c_k(\rho)|=
2d\sum_{i=1}^{m}\abs{c_i(\rho)}$ for $d=2m+1$,
$C_{\ell_1}(\Phi(\rho))=
2d\sum_{i=1}^{m-1}\abs{c_i(\rho)}+dc_m(\rho)$ for $d=2m$.

From $\Phi$ is a incoherent, it holds that
$C_{\ell_2}(\rho)\geqslant C_{\ell_2}(\Phi(\rho))$. By the
definition of the $\ell_2$-norm coherence, we obtain
\begin{eqnarray*}
    C_{\ell_2}(\Phi(\rho))=\frac1d\sum_{k=1}^{d-1}\Abs{\sum_{i=1}^{d}\rho_{i,i\oplus
            k}}^2,\quad
    C_{\ell_2}(\Delta(\rho))=\frac{1}{d(d-1)}\Abs{\sum_{i\neq j}\rho_{i,j}}^2.
\end{eqnarray*}
Using the Cauchy-Schwarz inequality, then
\begin{eqnarray*}
\Abs{\sum_{i\neq
j}\rho_{i,j}}^2\leqslant\Pa{\sum_{k=1}^{d-1}\Abs{\sum_{i=1}^{d}\rho_{i,i\oplus
                k}}}^2\leqslant(d-1)\sum_{k=1}^{d-1}\Abs{\sum_{i=1}^{d}\rho_{i,i\oplus
                k}}^2.
\end{eqnarray*}
Thus, $C_{\ell_2}(\Phi(\rho))\geqslant C_{\ell_2}(\Delta(\rho))$.
\end{proof}

Based on the generalized Gel-Mann representation of the qutrit state, we can obtain the changes in the generalized Bloch vector of the quantum state before and after the action of the circulant quantum channel.

\begin{exam}
Any qutrit state can be represented as \cite{Zhang2024ps}
\begin{eqnarray*}
    \rho=\frac13\Pa{\I_3+\sqrt{3}\bsr\cdot\bG},
\end{eqnarray*}
where $\bsr=(r_1,\cdots,r_8)\in\bbC^8$, $\bG=(\bsG_1,\cdots,\bsG_8)$
is the vector of Gell-Mann matrices on $\bbC^3$. Then
$\Phi(\rho)=c_0\I_3+c_1\bsP_{\pi_{0}}+\bar{c}_1 \bsP_{\pi_{0}}^2$,
where
$c_0=\frac13,~c_1=\frac{1}{3\sqrt{3}}(r_1+r_4+r_6-\mathrm{i}(r_2-r_5+r_7))$.
It is easily seen that the action of $\Phi$ is
\begin{eqnarray*}
    \Phi(\rho)=\frac13\Pa{\I_3+\sqrt{3}\bsr'\cdot\bG}=\Pa{
        \begin{array}{ccc}
            \frac{1}{3} & \frac{r_1+r_4+r_6-\mathrm{i}(r_2-r_5+r_7)}{3 \sqrt{3}} & \frac{r_1+r_4+r_6+\mathrm{i}(r_2-r_5+r_7)}{3 \sqrt{3}} \\
            \frac{r_1+r_4+r_6+\mathrm{i}(r_2-r_5+r_7)}{3 \sqrt{3}} & \frac{1}{3} & \frac{r_1+r_4+r_6-\mathrm{i}(r_2-r_5+r_7)}{3 \sqrt{3}} \\
            \frac{r_1+r_4+r_6-\mathrm{i}(r_2-r_5+r_7)}{3 \sqrt{3}} & \frac{r_1+r_4+r_6+\mathrm{i} (r_2-r_5+r_7)}{3 \sqrt{3}} & \frac{1}{3} \\
        \end{array}
    }
\end{eqnarray*}
where $\bsr'=\frac{1}{3}
(r_1+r_4+r_6,r_2-r_5+r_7,0,r_1+r_4+r_6,-\left(r_2-r_5+r_7\right),r_1+r_4+r_6,r_2-r_5+r_7,0).$

For the qutrit state, we calculate the $\ell_p$-norm coherence
$C_{\ell_p}(\rho)$ and the lower bound the $\ell_p$-norm coherence
$C_{\ell_p}(\Phi(\rho))$ and $C_{\ell_p}(\Delta(\rho))$,
\begin{eqnarray*}
C_{\ell_1}(\rho)=\frac{2}{\sqrt{3}}\Pa{\sqrt{r_1^2+r_2^2}+\sqrt{r_4^2+r_5^2}+\sqrt{r_6^2+r_7^2}},&~&C_{\ell_2}(\rho)=\frac{2}{3} \left(r_1^2+r_2^2+r_4^2+r_5^2+r_6^2+r_7^2\right)\\
C_{\ell_1}(\Phi(\rho))=\frac{2}{\sqrt{3}}\sqrt{(r_1+r_4+r_6)^2+(r_2-r_5+r_7)^2},&~&C_{\ell_2}(\Phi(\rho))=\frac{2}{9} \left((r_1+r_4+r_6)^2+(r_2-r_5+r_7)^2\right)\\
C_{\ell_1}(\Delta(\rho))=\frac{2}{\sqrt{3}}\Pa{r_1+r_4+r_6},&~&C_{\ell_2}(\Delta(\rho))=\frac{2}{9} (r_1+r_4+r_6)^2.
\end{eqnarray*}
Let us consider the qutrit state $\ket{\varphi}=\cos\theta\ket{0}+\frac{\sin\theta}{\sqrt{2}}e^{\mathbf{i}\phi}\ket{1}+\frac{\sin\theta}{\sqrt{2}}\ket{2}$, the vector $\bsr$ in the generalized Gell-Mann representation of $\rho=\out{\varphi}{\varphi}$ can be calculated by $r_i=\frac{\sqrt{3}}{2}\Tr{\rho\bsG_i}$. 
Let $\phi =\frac{\pi}{6}$, we calculated $\ell_p$-norm coherence
$C_{\ell_p}(\rho)$ and its lower bounds and compared them in Figure
\ref{fig}.
\end{exam}

\begin{figure}[htbp]
    \centering
    \subfigure[]{\includegraphics[width=8cm]{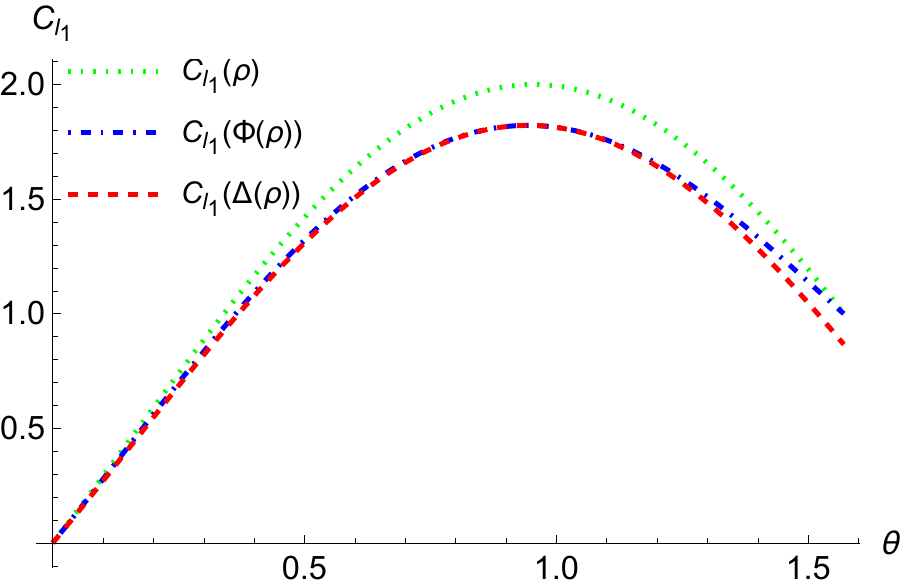}}
    \quad
    \subfigure[]{\includegraphics[width=8cm]{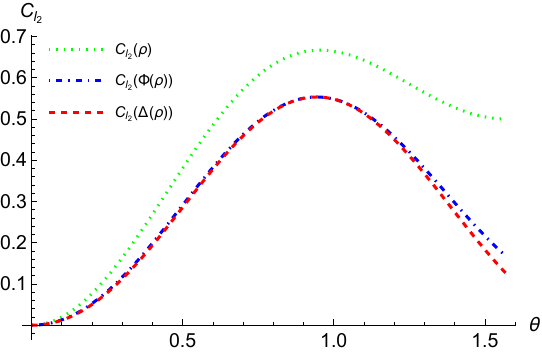}}
    \caption{The estimation of the coherence of $\rho$ by using $\ell_p$-norm coherence measures for $p=1,2$.}
    \label{fig}
\end{figure}

From the above discussion in this section, applying this class of
channels to the theory of Bargmann invariants recovers and
reformulates known foundational results. Crucially, it does so from
a unified operational channel perspective, offering a fresh
viewpoint on the interplay between collective phase information and
quantum correlations. We also see that circulant quantum channels
constitute a well-defined, highly symmetric set of "free" or
"non-resourceful" operations within a resource theory of coherence,
characterized by their covariance under cyclic shifts. This provides
a refined framework for probing quantum resources.

\section{Circulant quantum channels on bipartite systems}
\label{sect:bipartexp}

In this section, we study the circulant channel on bipartite
systems. Let $d_A \leqslant d_B$, and given that the orthonormal
bases of subsystems $\bbC^{d_A}$ and $\bbC^{d_B}$ are
$\set{\ket{i}}_{i=1}^{d_A}$ and $\set{\ket{j}}_{j=1}^{d_B}$
respectively. For the bipartite quantum system $\mathbb{C}^{d_A}
\otimes \mathbb{C}^{d_B}$, the density matrix of $\rho$ is given by
$$\rho_{AB} = \sum_{i,j,i',j'} \rho_{ij,i'j'} |ij\rangle \langle
i'j'|.$$ When $\rho_{AB}$ represents a separable state, it admits
the decomposition $\rho_{AB} = \sum_{i} p_{i} \rho_A^i \otimes
\rho_B^i$ with $\sum_ip_i=1$. Otherwise, $\rho_{AB}$ is an entangled
state.

For $\bsE_{ij}=\out{i}{j}\in \lin{\bbC^d}$ with orthonormal bases
$\set{\ket{i}}_{i=1}^{d}$, it holds that
\begin{eqnarray}\label{th3.1}
    \Phi(\bsE_{ij})=\frac1d \bsP_{\pi_0}^r, ~r=j\ominus i.
\end{eqnarray}
By Corollary~\ref{t1} and Eq.~\eqref{th3.1}, we derive the explicit
form of the channel $\Phi_A\ot\mathsf{id}_{B}$ and $\Phi_A\ot\Phi_B$ on
bipartite systems $\bbC^{d_A} \otimes \bbC^{d_B}$. Where $\sf{id}$
denote the identical channel on operator spaces.

As a specific application, we provide a complete characterization of
the entanglement-breaking property for the uniform-weight circulant
channel. This result, derived from the channel's distinctive
spectral properties, stands as a clear and nontrivial finding.

\begin{thrm}\label{thrm5.1}
For any bipartite operator $\bsX_{AB}$ acting on $\bbC^{d_A} \otimes
\bbC^{d_B}$, we have
\begin{eqnarray*}
\Phi_A\ot \mathsf{id}_B(\bsX_{AB})=\frac{1}{d_A}\sum_{k=0}^{d_A-1}
\bsP^{k}_{\pi_0}\ot\Ptr{A}{\bsX_{AB}\Pa{\bsP^{-k}_{\pi_0}\ot\I_B}}
\end{eqnarray*}
\end{thrm}

\begin{proof}
We decompose $\bsX$ with respect to two bases
$\set{\ket{i}}_{i=1}^{d_A}$ and $\set{\ket{j}}_{j=1}^{d_B}$ as
$$
\bsX_{AB} = \sum_{i,i'=1}^{d_A} \sum_{j,j'=1}^{d_B}
x_{ij,i'j'}\out{ij}{i'j'}.
$$
Then
\begin{eqnarray*}
&&\Phi_A\ot \mathsf{id}_B(\bsX_{AB})=\sum_{i,i'=1}^{d_A} \sum_{j,j'=1}^{d_B} x_{ij,i'j'}\Phi(\out{i}{i'})\ot\out{j}{j'}\\
&&=\sum_{i=1}^{d_A} \sum_{j,j'=1}^{d_B} x_{ij,ij'}\Phi(\out{i}{i})\ot\out{j}{j'}+\sum_{i=1}^{d_A} \sum_{j,j'=1}^{d_B} x_{ij,(i\oplus1)j'}\Phi(\out{i}{i\oplus1})\ot\out{j}{j'}+\cdots\\
&&~~~+\sum_{i=1}^{d_A} \sum_{j,j'=1}^{d_B} x_{ij,(i\oplus(d_A-1))j'}\Phi(\out{i}{i\oplus(d_A-1)})\ot\out{j}{j'}\\
&&=\frac{1}{d_A}\I_A\ot\Pa{\sum_{i=1}^{d_A} \sum_{j,j'=1}^{d_B} x_{ij,ij'}\out{j}{j'}}+\frac{1}{d_A}\bsP_{\pi_0}\ot\Pa{\sum_{i=1}^{d_A} \sum_{j,j'=1}^{d_B} x_{ij,(i\oplus1)j'}\out{j}{j'}}+\cdots\\
&&~~~+\frac{1}{d_A}\bsP^{d_A-1}_{\pi_0}\ot\Pa{\sum_{i=1}^{d_A} \sum_{j,j'=1}^{d_B} x_{ij,(i\oplus(d_A-1))j'}\ketbra{j}{j'}}\\
&&=\frac{1}{d_A}\Big(\I_A\ot\Ptr{A}{\bsX_{AB}}+\bsP_{\pi_0}\ot\Ptr{A}{\bsX_{AB} \Pa{\bsP_{\pi^{-1}_0}\ot\I_B}}+\cdots\\
&&~~~+\bsP^{d_A-1}_{\pi_0}\ot \Ptr{A}{\bsX_{AB} \Pa{\bsP^{d_A-1}_{\pi^{-1}_0}\ot\I_B}}\Big)\\
&&=\frac{1}{d_A}\sum_{k=0}^{d_A-1}
\bsP^{k}_{\pi_0}\ot\Ptr{A}{\bsX_{AB}\Pa{\bsP^{-k}_{\pi_0}\ot\I_B}},
\end{eqnarray*}
due to
\begin{eqnarray*}
\Ptr{A}{\bsX_{AB}\Pa{\bsP^{-k}_{\pi_0}\ot\I_B}}&=&\sum_{i,i',p=1}^{d_A} \sum_{j,j',q=1}^{d_B} x_{ij,i'j'}\Ptr{A}{\ketbra{ij}{i'j'}\cdot \ketbra{(p\oplus k)q}{p q}}\\
        &=&\sum_{i,p=1}^{d_A} \sum_{j,q=1}^{d_B} x_{ij,(p\oplus k)q}\Ptr{A}{\out{ij}{pq}}\\
        &=&\sum_{p=1}^{d_A} \sum_{j,q=1}^{d_B} x_{pj,(p\oplus k)q}\out{j}{q}.
\end{eqnarray*}
\end{proof}

In this case, the output $\Phi_A \otimes \mathsf{id}_B(\bsX_{AB})$
is a block circulant matrix. If the quantum circulant channel acts
locally on two subsystems of a bipartite system simultaneously, we
can also obtain the explicit expression of the channel. At this
point, the output $\Phi_A \otimes \Phi_B(\bsX_{AB})$ is also a block
circulant matrix where each block is itself a circulant matrix.

\begin{thrm}
    For any bipartite operator $\bsX_{AB}$ acting on $\mathbb{C}^{d_A} \otimes \mathbb{C}^{d_B}$, we have
    \begin{eqnarray*}
        \Phi_A\ot \Phi_B(\bsX_{AB})=\frac{1}{d_Ad_B}\sum_{r_A=0}^{d_A-1}\sum_{r_B=0}^{d_B-1}\Tr{\bsX_{AB}\Pa{\bsP^{-r_A}_{\pi_0}\ot\bsP^{-r_B}_{\pi_0}}}\bsP^{r_A}_{\pi_0}\ot \bsP^{r_B}_{\pi_0}.
    \end{eqnarray*}
\end{thrm}

\begin{proof}
    From Eq. (\ref{th3.1}), we have
    \begin{eqnarray*}
        \Phi_A\ot \Phi_B(\bsX_{AB})&=&\sum_{i,i'=1}^{d_A} \sum_{j,j'=1}^{d_B} x_{ij,i'j'}\Phi(|i\rangle\langle i'|)\ot\Phi(\ketbra{j}{j'})\\
        &=&\frac{1}{d_Ad_B}\sum_{r_A=0}^{d_A-1} \sum_{r_B=0}^{d_B-1}\sum_{i=1}^{d_A} \sum_{j=1}^{d_B} x_{ij,(i\oplus r_A)(j\oplus r_B)}\bsP_{\pi_0}^{r_A}\ot \bsP_{\pi_0}^{r_B}\\
        &=&\frac{1}{d_Ad_B}\sum_{r_A=0}^{d_A-1}\sum_{r_B=0}^{d_B-1}\Tr{\bsX_{AB}\Pa{\bsP^{-r_A}_{\pi_0}\ot\bsP^{-r_B}_{\pi_0}}}\bsP^{r_A}_{\pi_0}\ot \bsP^{r_B}_{\pi_0},
    \end{eqnarray*}
    where $\Tr{\bsX_{AB}\Pa{\bsP^{-r_A}_{\pi_0}\ot\bsP^{-r_B}_{\pi_0}}}=\sum_{i=1}^{d_A} \sum_{j=1}^{d_B} x_{ij,(i\oplus r_A)(j\oplus r_B)}$.
\end{proof}

For any $2\ot d_B$ states $\rho_{AB}$, due to
$\bsP_{\pi_0}^\t=\bsP_{\pi_0}$ for $\pi_0\in S_2$, we have
$(\Phi_A\ot\mathsf{id}_B (\rho_{AB}))^{\Gamma_A}=\Phi_A\ot\mathsf{id}_B
(\rho_{AB})\geqslant0$ and
$(\Phi_A\ot\Phi_B(\rho_{AB}))^{\Gamma_A}=\Phi_A\ot\mathsf{id}_B
(\rho_{AB})\geqslant0$, where $\Gamma_A$ denote partial-transpose with respect to the subsystem $A$. Thus, both $\Phi_A\ot\mathsf{id}_B
(\rho_{AB})$ and $\Phi_A\ot\Phi_B(\rho_{AB})$ are positive under
partial transpositions (PPT) \cite{Peres1996}. Moreover, according
to the PPT criteria, if $d_B=2,3$, then the output states is
separable. This indicates that in the qubit-qubit system and the
qubit-qutrit system, the local circulant channel completely
eliminates the entanglement between the two subsystems.

A channel $\cT \equiv \cT_A$ in which the entanglement is completely
erased like this is called an entanglement-breaking channel, which
means that $(\cT_A \otimes \text{id}_B)(\rho_{AB})$ is separable for
every bipartite state $\rho_{AB} \in \density{\mathbb{C}^{d_A}\ot
\mathbb{C}^{d_B}}$. In qudit systems, any quantum channel is
entanglement breaking if and only if its Choi representation is
separable \cite{Horodecki2003}. Next, we will prove that the channel
$\Phi$ is entanglement-breaking channel.

\begin{thrm}\label{th:entbreaking}
In the family $\Set{\Phi_{\boldsymbol{\lambda}}:\boldsymbol{\lambda}\in\bbR^d \text{ is a
probability vector}}$ of circulant quantum channels, there is only
one channel $\Phi$ (corresponding to $\Phi_{\boldsymbol{\lambda}}$ for
$\boldsymbol{\lambda}=(1/d,\ldots,1/d)$) is entanglement-breaking. In other
words, $\Phi_{\boldsymbol{\lambda}}$ is entanglement-breaking if and only if
$\boldsymbol{\lambda}=(1/d,\ldots,1/d)$.
\end{thrm}

\begin{proof}
(1) We first show that the Choi-Jami{\l}kowski representation of the
channel $\Phi_\lambda$ is given by
\begin{eqnarray*}
J(\Phi_{\boldsymbol{\lambda}})
=(\bsF\ot\bar\bsF)\Pa{\sum^{d-1}_{i,j=0}\alpha_{i-j}(\boldsymbol{\lambda})\out{ii}{jj}}(\bsF\ot\bar\bsF)^\dagger,
\end{eqnarray*}
where
$$
\alpha_\mu(\boldsymbol{\lambda}):=\frac1d\sum^{d-1}_{k=0}\lambda_k\omega^{k\mu}.
$$
The set of all eigenvalues of the partial-transposed Choi-Jami{\l}kowski
representation $J(\Phi_{\boldsymbol{\lambda}})^\Gamma$, with respect to the second
subsystem, is given by
\begin{eqnarray}
\Set{[\alpha_0(\boldsymbol{\lambda})]_{(d)}, [\pm
\abs{\alpha_\mu(\boldsymbol{\lambda})}]_{(d)}\text{ for each }\mu\neq0}.
\end{eqnarray}
Indeed,
\begin{eqnarray*}
J(\Phi_\lambda) &=&
\sum^{d-1}_{k=0}\lambda_k\vec(\bsP^k_{\pi_0})\vec(\bsP^k_{\pi_0})^\dagger\\
&=&
(\bsF\ot\bar\bsF)\Pa{\sum^{d-1}_{k=0}\lambda_k\vec(\Omega^k)\vec(\Omega^k)^\dagger}
(\bsF\ot\bar\bsF)^\dagger,
\end{eqnarray*}
where the second equality holds because of
Eq.~\eqref{eq:Omega}. Using
$\Omega^k=\sum^{d-1}_{i=0}\omega^{ki}\proj{i}$ and
$\vec(\Omega^k)=\sum^{d-1}_{i=0}\omega^{ki}\ket{ii}$, we have
\begin{eqnarray*}
\sum^{d-1}_{k=0}\lambda_k\vec(\Omega^k)\vec(\Omega^k)^\dagger =
\sum^{d-1}_{i,j=0}\Pa{\frac1d\sum^{d-1}_{k=0}\lambda_k\omega^{k(i-j)}}\out{ii}{jj}=\sum^{d-1}_{i,j=0}\alpha_{i-j}(\boldsymbol{\lambda})\out{ii}{jj},
\end{eqnarray*}
whose partial-transpose with respect to the second subsystem (denote
by $\Gamma$) is given by
\begin{eqnarray*}
\Pa{\sum^{d-1}_{k=0}\lambda_k\vec(\Lambda^k)\vec(\Lambda^k)^\dagger}^\Gamma
=\sum^{d-1}_{i,j=0}\alpha_{i-j}(\boldsymbol{\lambda})\out{ij}{ji}=:\bsM.
\end{eqnarray*}
Then
\begin{eqnarray*}
\bsM\bsM^\dagger =\sum^{d-1}_{i,j=0}\abs{
\alpha_{i-j}(\boldsymbol{\lambda})}^2\out{ij}{ij}\Longrightarrow\abs{\bsM}=\sum^{d-1}_{i,j=0}\abs{
\alpha_{i-j}(\boldsymbol{\lambda})}\out{ij}{ij}
\end{eqnarray*}
Therefore the set of eigenvalues of $\abs{\bsM}$ is given by
\begin{eqnarray*}
\set{\abs{\alpha_{i-j}(\boldsymbol{\lambda})}:i,j=0,\ldots,d-1}.
\end{eqnarray*}
Thus the set of eigenvalues of $\bsM$ is given by
\begin{eqnarray*}
\Set{(\alpha_0(\boldsymbol{\lambda})=1/d)_{(d)},\Pa{\pm\abs{\alpha_\mu(\boldsymbol{\lambda})}}_{(d)}\text{ for
each }\mu\neq0}.
\end{eqnarray*}
(2) In the family
$\set{\Phi_{\boldsymbol{\lambda}}:\boldsymbol{\lambda}\in\bbR^d\text{
is a probability vector}}$, we show that $\Phi$ is the unique
entanglement-breaking channel. Indeed, from the above reasoning, we
see that $\bsM$ is positive semi-definite if and only if
$\abs{\alpha_\mu}=0$ for all $\mu\neq0$. This implies that
$\alpha_\mu=0$ for all $\mu\neq0$. The only solution is
$\lambda_k=\frac1d$ for all $k=0,1,\ldots,d-1$. In summary,
$J(\Phi_{\boldsymbol{\lambda}})^\Gamma\geqslant0$ if and only if
$\lambda_k=\frac1d$ for all $k=0,1,\ldots,d-1$. In other words,
$\Phi_{\boldsymbol{\lambda}}$ is not entanglement-breaking if and
only if $\boldsymbol{\lambda}\neq(1/d,\ldots,1/d)$. Through the
relationship between Klaus representation and Choi representation,
we can calculate the expression of $J(\Phi)$ \cite{Watrous}. Indeed,
via $\bsP_{\pi_0}=\bsF\Omega\bsF^\dagger$,
\begin{eqnarray*}
J(\Phi)&=&
\frac1d\sum^{d-1}_{k=0}\vec(\bsP^k_{\pi_0})\vec(\bsP^k_{\pi_0})^\dagger
=
\frac1d\sum^{d-1}_{k=0}\vec(\bsF\Omega^k\bsF^\dagger)\vec(\bsF\Omega^k\bsF^\dagger)^\dagger\\
&=&
\Pa{\bsF\ot\bar\bsF}\Pa{\frac1d\sum^{d-1}_{k=0}\vec(\Omega^k)\vec(\Omega^k)^\dagger}\Pa{\bsF\ot\bar\bsF}^\dagger,
\end{eqnarray*}
where the last equality is due to the identity
$\vec{(\bsA\bsX\bsB)}=(\bsA\ot\bsB^{\T})\vec{(\bsX)}$. In fact,
after simple calculations, we can obtain the following result
\begin{eqnarray}
\frac1d\sum^{d-1}_{k=0}\vec(\Omega^k)\vec(\Omega^k)^\dagger =
\sum^{d-1}_{i=0}\proj{ii}= \sum^{d-1}_{i=0}\proj{i}\ot\proj{i},
\end{eqnarray}
which is separable operator. Since $J(\Phi)$ is locally unitary to
separable operator, it follows that $J(\Phi)$ is also separable.
This leads to the fact that $\Phi$ is entanglement-breaking.
\end{proof}

\begin{remark} We can also derive the expression of $J(\Phi)$
by the definition \cite{Watrous}. The Choi-representation of the
circulant channel $\Phi$ is given by
$J(\Phi):=\Phi\ot\mathsf{id}(\mathrm{\vec}(\I_d)\mathrm{\vec}(\I_d)^\dagger)$,
where $\mathrm{\vec}(\I_d)=\sum_{i=1}^{d}\ket{ii}$. By the Theorem
\ref{thrm5.1}, it follows that
\begin{eqnarray*}
\Phi\ot\mathsf{id}(\mathrm{\vec}(\I_d)\mathrm{\vec}(\I_d)^\dagger)&=&\frac{1}{d}\sum_{k=0}^{d-1} \bsP^{k}_{\pi_0}\ot\Ptr{A}{\mathrm{\vec}(\I_d)\mathrm{\vec}(\I_d)^\dagger\Pa{\bsP^{-k}_{\pi_0}\ot\I_B}}\\
&=& \frac{1}{d}\sum_{k=0}^{d-1}\bsP^{k}_{\pi_0}\ot \sum_{i,j=1}^{d}\Ptr{A}{\ketbra{ii}{jj} \Pa{\sum_{p=1}^{d}\ketbra{p}{\pi_0^k(p)}\ot\I_B}}\\
&=&\frac{1}{d}\sum_{k=0}^{d-1}\bsP^{k}_{\pi_0}\ot\sum_{i,j,p=1}^{d}\Ptr{A}{\Pa{\ketbra{i}{j} \ketbra{p}{\pi_0^k(p)}}\ot\ketbra{i}{j}}\\
&=&\frac{1}{d}\sum_{k=0}^{d-1}\bsP^{k}_{\pi_0}\ot \bsP^{k}_{\pi_0}.
\end{eqnarray*}
Finally, the following identity holds:
\begin{eqnarray*}
\Pa{\bsF\ot\bar\bsF}\Pa{\frac1d\sum^{d-1}_{k=0}\vec(\Omega^k)\vec(\Omega^k)^\dagger}\Pa{\bsF\ot\bar\bsF}^\dagger
=
(\bsF\ot\bsF)\Pa{\frac1d\sum^{d-1}_{k=0}\Omega^k\ot\Omega^k}(\bsF\ot\bsF)^\dagger
\end{eqnarray*}
which is equivalent to
\begin{eqnarray*}
\Pa{\bsF\ot\bar\bsF}\Pa{\sum^{d-1}_{i=0}\proj{ii}}\Pa{\bsF\ot\bar\bsF}^\dagger
= (\bsF\ot\bsF)\Pa{
\sum^{d-1}_{\stackrel{i,j=0}{i+j\equiv0(\!\!\!\mod
d)}}\proj{ij}}(\bsF\ot\bsF)^\dagger.
\end{eqnarray*}
Therefore $J(\Phi)$ is separable, and thus $\Phi$ is
entanglement-breaking.
\end{remark}

\begin{cor}
For any bipartite entangled state
$\rho_{AB}\in\density{\bbC^{d_A}\ot\bbC^{d_B}}$, it holds that
$$
(\Phi_{\boldsymbol{\lambda}_A}\ot
\Phi_{\boldsymbol{\lambda}_B})(\rho_{AB})
$$
is still entangled whenever both probability vectors
$\boldsymbol{\lambda}_A\in\bbR^{d_A}$ and
$\boldsymbol{\lambda}_B\in\bbR^{d_B}$ are not uniform.
\end{cor}

\begin{proof}
The proof is easily obtained from immediately from
Theorem~\ref{th:entbreaking}.
\end{proof}

\section{Conclusions}\label{sect:concl}

We investigate a subclass of mixed-permutation channels---circulant
quantum channels---and characterize their fundamental properties.
These channels not only enable tighter bounds for $\ell_p$-norm
$(p=1,2)$ coherence measures but also resolve the characterization
problem for Bargmann invariant sets. Furthermore, we analyzed the
action of the circulant quantum channel in the bipartite system and
revealed a notable feature of this channel: as an
entanglement-breaking channel, it can completely eliminate quantum
entanglement at a small cost, requiring only a small number of Kraus
operators.

Given the pivotal role of circulant Gram matrices in characterizing
Bargmann invariants, we propose exploring whether other quantum
circulant channel families exist. Owing to the distinctive
spectral and structural properties of circulant matrices, we
anticipate significant applications of quantum circulant channels
across diverse domains.

\subsection*{Acknowledgement}

This work is supported by Zhejiang Provincial Natural Science
Foundation of China under Grants No. LZ23A010005.\\~\\
\textbf{Data Availability Statement} No data are associated in the
manuscript.\\~\\
\textbf{Author Contributions Statement} B.X. and L.Z. wrote the main
manuscript text and B.X. prepared figures 1-2. All authors reviewed
the manuscript.



\end{document}